\def\year{2019}\relax
\documentclass[letterpaper]{article} 
\usepackage{aaai19}  
\usepackage{times}  
\usepackage{helvet}  
\usepackage{courier}  
\usepackage{url}  
\usepackage{graphicx}  
\usepackage{amsmath}
\usepackage{amssymb}
\newcommand{\opt}{{\sf Opt}}

\newcommand{\alg}{{\sf Alg}}

\newcommand{\prob}[1]{\operatorname{Pr}\left[#1\right]}
\newcommand{\ex}[2][]{\operatorname{E}_{#1}\left[#2\right]}

\newenvironment{proof}{\noindent{\bf Proof : \ }}{\hfill$\Box$\par\medskip}

\newtheorem{theorem}{Theorem}
\newtheorem{corollary}[theorem]{Corollary}
\newtheorem{lemma}[theorem]{Lemma}

\newtheorem{definition}[theorem]{Definition}

\begin{document}
%
\title{Online Pandora's Boxes and Bandits}
\author{ Hossein Esfandiari \thanks{ {\em esfandiari@google.com}}\\Google Research 
	\\ New York, NY
	\And MohammadTaghi HajiAghayi \thanks{{\em hajiagha@cs.umd.edu}}  \\University of Maryland\\ 
	Department of Computer Science
	\\ College Park, MD 
	\And
	Brendan Lucier \thanks{{\em brlucier@microsoft.com} } \\Microsoft Research
	\\ Cambridge, MA  
	\AND
	Michael Mitzenmacher\thanks{{\em michaelm@eecs.harvard.edu}}\\ Harvard University \\ School of Engineering and Appplied Science\\
	Cambridge, MA
}

\maketitle

\begin{abstract}
We consider online variations of the Pandora's box problem \cite{weitzman79},
a standard model for understanding issues related to the cost of acquiring
information for decision-making.  Our problem generalizes both the classic 
Pandora's box problem and the prophet inequality framework.
Boxes are presented online, each with a random value and cost drawn
jointly from some known distribution. Pandora chooses online whether to
open each box given its cost, and then chooses irrevocably whether to keep 
the revealed prize or pass on it.  We aim
for approximation algorithms against adversaries that can choose the
largest prize over any opened box, and use optimal offline policies to
decide which boxes to open (without knowledge of the value
inside)\footnote{See section~\ref{sec:reduction} for formal definitions.}.  
We consider variations where Pandora can collect
multiple prizes subject to feasibility constraints, such as cardinality, 
matroid, or knapsack constraints.  
We also consider
variations related to classic 
multi-armed bandit problems from reinforcement learning. Our results
use a reduction-based framework where we separate the issues of the
cost of acquiring information from the online decision process of which prizes 
to keep. Our work shows that in many scenarios, Pandora can achieve a 
good approximation to the best possible performance.
\end{abstract}

\section{Introduction}

Information learning costs play a large role in a variety of
markets and optimization tasks.  For example, in the academic
job market, obtaining information about a potential match is a costly 
investment for both sides of the market.  Conserving on information costs is an 
important component of efficiency in such settings.

A classic model for information learning costs is the Pandora's
box problem, attributed to Weitzman \cite{weitzman79}, which has 
the following form.  Pandora has $n$ boxes,
where the $i$th box contains a prize of value $v_i$ that has a known
cumulative distribution function $F_i$.  It costs $c_i$ to open the
$i$th box and reveal the actual value $v_i$.  Pandora may open as
many boxes as she likes, in any order.  The payoff is the
maximum-valued prize, minus the cost of the opened boxes.  That is,
if $S$ is the subset of opened boxes, then the payoff Pandora seeks
to maximize is
$$\max_{i \in S} v_i - \sum_{i \in S} c_i.$$
The Pandora's box problem incorporates two key decision
aspects: the ordering for opening boxes, and when to stop.
It has been proposed for applications
such as buying or selling a house and searching for a job.

The original Pandora's box problem has a simple and elegant solution.
The {\em reservation price} $v_i^*$ associated with an unopened box $i$ is
the value for which Pandora would be indifferent taking a prize with
that value and opening box $i$.  That is,
\begin{align*}
v_i^* & = \inf\{y : y \geq -c_i + \ex{\max{v_i,y}}\} \\
& = \inf\{y : c_i \geq \ex{\max{v_i-y,0}}\}.
\end{align*}
This result says that if Pandora is allowed to choose the ordering,
Pandora should keep opening boxes in the order of decreasing
reservation price, but should stop searching when the largest prize
value obtained exceeds the reservation price of all unopened boxes.
An alternative proof to Weitzman's proof \cite{weitzman79} of this
was recently provided by Kleinberg, Waggoner, and
Weyl~\cite{KWW-16}, who also present additional applications,
including to auctions. Very recently Singla~\cite{Singla18}
generalizes the approach of Kleinberg et al.~\cite{KWW-16} for more
applications in offline combinatorial problems such as matching, set
cover, facility location, and prize-collecting Steiner tree.

In other similar problems, the ordering is chosen adversarially and
adaptively.  For example, in the {\em prophet inequality} setting
first introduced in 1977 by Garling, Krengel, and
Sucheston~\cite{krengel1978semiamarts,krengel1977semiamarts}, the boxes have
no cost, and the prize distributions are known, but the
decision-maker has to decide after each successive box whether to
stop the process and keep the corresponding prize; if not, the prize
cannot be claimed later.  It is known that there exists a
threshold-based algorithm that in expectation obtains a prize value
within a factor of two of the expected maximum prize (and the factor of two
is tight) \cite{krengel1978semiamarts,krengel1977semiamarts}.
There have subsequently been many generalizations of the prophet inequality setting,
especially to applications in online auctions (see
e.g.~\cite{hajiaghayi2007automated,alaei2012online,hajiaghayi2004adaptive,alaei2014bayesian,yan2011mechanism,KW-STOC12,babaioff2007matroids,Lachish-FOCS14,FSZ-SODA15,GM-SODA08,KorulaPal-ICALP09,MY-STOC11,KMT-STOC11,kesselheim2013optimal,GS-IPCO17,esfandiari2015prophet,duetting2017prophet}).

Another related and well-studied theme includes {\em multi-armed
bandit} problems and more generally reinforcement learning (see,
e.g., \cite{GJ74}).  In this setting, each ``box''
corresponds to a strategy, or arm, that has a payoff in each round.
An online algorithm chooses one arm from a set of arms in each round
over $n$ rounds.  Viewed in the language of selection problems, this 
translates to a feasibility constraint on the set of boxes that can be opened. 
Multi-armed bandit problems have applications
including online auctions, adaptive routing, and the theory of
learning in games.

In this paper, we consider a class of problems that combine
the cost considerations of Pandora's box with the online nature of
prophet inequality problems.  Again boxes are presented online, here
with random values and costs drawn jointly from some distribution.
Pandora chooses online whether to open each box, and then whether to
keep it or pass on it.  We aim for approximation algorithms against
adversaries that can choose the largest prize over any opened box,
and use optimal offline policies in deciding which boxes to open, without
knowledge of the value inside.  We consider variations where
Pandora can collect multiple prizes subject to sets of constraints.
For example, Pandora may be able to keep at most $k$ prizes, the
selected prizes must form an independent set in a matroid, or the prizes
might have associated weights that form a knapsack constraint. 
We also introduce
variations related to classic multi-armed bandit problems and
reinforcement learning, where there are feasibility constraints on the set of
boxes that can be opened. Our work shows that in many scenarios even
without the power of ordering choices, Pandora can achieve a good
approximation to the best possible performance.

Our main result is a reduction from this general class of problems, which we
refer to as online Pandora's box problems, to the problem of finding
threshold-based algorithms for the associated prophet inequality problems 
where all costs are zero.  Our reduction is constructive, and results in
a polynomial-time policy, given a polynomial-time algorithm for constructing
thresholds in the reduced problem.  
We first describe the reduction in Section~\ref{sec:reduction}.
Then in Section~\ref{sec:applications}, we show how to use known results from the prophet inequality literature
to directly infer solutions to online Pandora's box problems.  Finally, in Section~\ref{sec:multiarm},
we establish an algorithm for a multi-armed bandit variant of the online Pandora's box problem,
by proving a novel multi-armed prophet inequality.

\section{Pandora's Boxes Under General Constraints}
\label{sec:reduction}
In this section we consider a very general version of an online Pandora's box problem, with the goal of showing that, if there is a suitable corresponding prophet inequality algorithm, we can use it in a way that yields good approximation ratios for the Pandora's box problem.  We define the problem as follows. There is a sequence of boxes that arrive online, in an order chosen by an adversary (i.e., worst case order). Each box has a cost $c_i$, a value $v_i$, and a type $t_i$. The tuple $(c_i, v_i, t_i)$ is drawn from a joint distribution $F_i$. The distributions are known in advance. When a box is presented, we observe its type $t_i$. We can then choose whether to open the box. We note that, given the type $t_i$, $v_i$ and $c_i$ have conditional distributions depending on the type $t_i$.  There is a set of constraints dictating what combinations of boxes can be opened; these constraints can depend on the indexes of the boxes and their types. If we open the box, then $v_i$ and $c_i$ are revealed, and we pay $c_i$ for opening the box. We must then choose (irrevocably) whether to keep and collect $v_i$.  There is also a set of constraints dictating what combinations of values can be kept; these constraints can depend on the indexes of the boxes and their types. We indicate the set of opened boxes by $S\subseteq\{1,\dots,n\}$ and the set of of kept boxes by $R\subseteq S$.
The goal is to maximize the expected utility $\ex{\sum_{i\in R} v_i - \sum_{i \in S} c_i}$.

One might want to consider an adversary
that obtains $E[\max_{R, S} \sum_{i \in R} v_i - \sum_{i \in S} c_i]$, that is, a fully clairvoyant adversary. However, it is not possible to provide any competitive algorithm against such an adversary even for the simple classical Pandora's box problem.\footnote{Consider the following example with $n$ identical (and independent) boxes where we have no constraint on opening boxes and can accept exactly one box at the end. The value of each box is $0$ with probability $1/2$ and $2$ with probability $1/2$;  the cost of each box is $1$. Note that the cost of each box is equal to its expected value. Hence, the expected utility of any online algorithm is upper bounded by $0$. However, with probability $1-\frac{1}{2^n}$ at least one of the boxes has a value $2$. A fully clairvoyant adversary only opens the box with value $2$ and obtain a utility $2\times(1-\frac{1}{2^n})-1=1-\frac{1}{2^{n-1}}$. This is a positive utility when $n\geq 2$.}

We denote the expected utility of an algorithm $\alg$ by
$u_{\alg}$. We compare our algorithms against the (potentially
exponential time) optimal offline algorithm $\opt$ that maximizes the
expected utility. Specifically, $\opt$ can see all the types $t_i$ of
all the boxes, and $\opt$ can choose to open boxes in any order.
However, $\opt$ does not learn the resulting cost and value $c_i$
and $v_i$ for a box until it is opened. $\opt$ iteratively and adaptively opens boxes and at the
end chooses a subset of opened boxes to keep. Of course $\opt$ must
respect the constraints on opened boxes and kept boxes.  

We first prove some fundamental lemmas that capture important structure for this Pandora's box problem. We then use these lemmas to provide a strong connection between the online Pandora's box problem and prophet inequalities. Our results allow us to translate several prophet inequalities algorithms, such as prophet inequalities under capacity constraints, matroid constraints, or knapsack constraints, to algorithms for online Pandora's box algorithms under the same constraints.

\subsection{Fundamental Lemmas}

Our lemmas allow us flexibility in considering the distribution of costs, and show how we can 
preserve approximation ratios.  

\begin{definition}
Let $F_1,\dots,F_n$, and $F'_1,\dots,F'_n$ be two sequences of boxes.  Denote the outcomes of $F_i$ and $F'_i$ by $(v_i,c_i,t_i)$ and $(v'_i,c'_i,t'_i)$ respectively.  We say the two sequences are \emph{cost-equivalent} if (a) they can be coupled so that for all $i$ we have $v_i=v'_i$ and $t_i=t'_i$, and (b) $\ex{c_i \ |\ t_i}=\ex{c'_i\ |\ t_i}$ for all $t_i$.
\end{definition}

\begin{lemma}\label{lm:base}
	Let $F_1,\dots,F_n$, and $F'_1,\dots,F'_n$ be two cost-equivalent sequences of boxes.  
	Let $\alg'$ be an online (resp., offline) algorithm that achieves an expected utility $u_{\alg'}$ on boxes $F'_1,\dots,F'_n$. There exists an online (resp., offline) algorithm $\alg$ that achieves the same expected utility on boxes $F_1,\dots,F_n$.
\end{lemma}
\begin{proof}
We will first suppose that $\alg'$ is online, so that the order of arrival is predetermined and the types are revealed online.
	We will define algorithm $\alg$ using the run of algorithm $\alg'$ on a simulated set of boxes $F'_1,\dots,F'_n$. When $\alg'$ attempts to open a box $F'_i$ we do the following. We open $F_i$, and let $(v_i,c_i,t_i)$ be the outcome. Then we draw a triple $(v'_i,c'_i,t'_i)$ from $F'_i$, conditioning on $v'_i=v_i$ and $t'_i=t_i$, and report it to $\alg'$. $\alg$ then opens the box if and only if $\alg'$ does, and likewise keeps the box if and only if $\alg'$ does.
	
	Let $Y_i$ be a binary random variable that is $1$ if $\alg'$ opens box $i$ and $0$ otherwise. Also, let $X_i$ be a binary random variable that is $1$ if $\alg'$ keeps box $i$ and $0$ otherwise.  Note that for any particular $i$, at the time that the algorithm decides about $Y_i$, $c_i$ is unknown to the algorithm.  Moreover, $c_i$ may be correlated with $t_i$, but is independent of all observations of the algorithm from prior rounds.  Therefore, after conditioning on $t_i$, $Y_i$ is independent of $c_i$. We have
	\begin{align*}
		&u_{\alg'}=
		 \ex{\sum_{i=1}^n X_i v'_i - \sum_{i=1}^n Y_i c'_i} \hspace{0.9cm}\text{Definition of utility}\\
		&=\ex{\sum_{i=1}^n X_i v'_i - \sum_{i =1}^n \ex{Y_i c'_i}} \hspace{0.9cm}\text{Linearity of Exp.}\\
		&=\ex{\sum_{i=1}^n X_i v'_i - \sum_{i =1}^n \ex[t'_i]{\ex{Y_i c'_i\ |\ t'_i}}}\hspace{0.2cm}\text{Draw $t'_i$ first}\\
		&=\ex{\sum_{i=1}^n X_i v'_i - \sum_{i =1}^n \ex[t'_i]{Y_i{\ex{ c'_i\ |\ t'_i}}}} \hspace{0.1cm}\text{ $Y_i$ indep.\ of $c'_i|t'_i$}\\
		&=\ex{\sum_{i=1}^n X_i v_i - \sum_{i =1}^n \ex[t_i]{Y_i\ex{c_i\ |\ t_i}}} \hspace{0.2cm}\text{ cost-equivalence}\\ 
		&=\ex{\sum_{i=1}^n X_i v_i - \sum_{i =1}^n \ex[t_i]{\ex{Y_i c_i\ |\ t_i}}} \hspace{0.1cm}\text{ $Y_i$ indep.\ of $c_i|t_i$}\\
\end{align*}
\begin{align*}
		&=\ex{\sum_{i=1}^n X_i v_i - \sum_{i =1}^n Y_i c_i}  \hspace{1.25cm}\text{Linearity of Exp.}\\
		&=u_{\alg}. \hspace{1.3cm}\text{$\alg$ opens and keeps the same sets as $\alg'$}
	\end{align*}
The case where $\alg'$ is an offline algorithm is similar.  The only difference is that the full profile of types is known to the algorithm in advance, and hence $Y_i$ and $X_i$ can depend on this profile.  We therefore fix the type profile, interpret variables $X_i$ and $Y_i$ as being conditioned on this realization of the types, interpret all expectations with respect to this conditioning, and the argument proceeds as before (noting that the distribution of $c'_i$ can depend on $t'_i$, but is independent of other types).  Note that this actually simplifies the chain of inequalities above, as the conditioning on $t'_i$ on the third line is trivial and unnecessary when $t'_i$ is fixed.
\end{proof}
 
\begin{lemma}\label{lm:approx}
	Let $F_1,\dots,F_n$, and $F'_1,\dots,F'_n$ be two cost-equivalent sequences of boxes.	
	Let $\alg'$ be an online (resp., offline) $\alpha$-approximation algorithm on boxes $F'_1,\dots,F'_n$. There exists an online (resp., offline) $\alpha$-approximation algorithm $\alg$ on boxes $F_1,\dots,F_n$.
\end{lemma}
\begin{proof}
	Let $\opt$ and $\opt'$ be the optimum (offline) algorithms for boxes $F_1,\dots,F_n$ and $F'_1,\dots,F'_n$ respectively. Let $\alg$ be the algorithm of Lemma \ref{lm:base} applied to $\alg'$.
	Moreover, note that Applying $\opt$ to Lemma \ref{lm:base} implies that there is some offline algorithm $\alg''$ on boxes $F'_1,\dots,F'_n$ such that $u_{\opt} = u_{\alg''}$.
	 We bound the approximation factor of $\alg$ as follows.
	\begin{align*}
		\frac{u_{\alg}}{u_{\opt}} &= \frac{u_{\alg'}}{u_{\opt}} &\text{By Lemma \ref{lm:base}}\\
		&= \frac{u_{\alg'}}{u_{\alg''}} &	u_{\opt} = u_{\alg''}	\\
		&\geq \frac{u_{\alg'}}{u_{\opt'}} &\text{definition of $\opt'$}\\
		&\geq \alpha. &\text{$\alg'$ is an $\alpha$-approximation algorithm}
	\end{align*} 
\end{proof}

Next we define the commitment Pandora's box problem on boxes $F^*_1 = (v^*_1,c^*_1,t^*_1),\dots,F^*_n = (v^*_n,c^*_n,t^*_n)$. Commitment Pandora's box is similar to the Pandora's box problem with the following two restrictions, that we refer to as freeness and commitment, respectively. 
\begin{itemize}
	\item {\bf Freeness:} Opening any box is free, i.e., for all $i$ we have $c^*_i=0$. 
	\item {\bf Commitment:} If a box $F^*_i$ is opened and the value $v^*_i$ is the maximum possible value of $F^*_i$, $v^*_i$ is kept. 
\end{itemize}
Note that the commitment constraint is without loss of generality for an online algorithm, but is a non-trivial restriction for an offline algorithm.

For each $i$, and for each type $t_i$, we define a threshold $\sigma_i$ so that we have $\ex{v_i-\min(v_i,\sigma_i)\ |\ t_i}=\ex{c_i\ |\ t_i}$. If $c_i=0$, we set $\sigma_i$ to the supremum of the $v_i$.  (It is possible to have $\sigma_i$ be infinity, with the natural interpretation.)  We define $F^*_i = (\min(v_i,\sigma_i),0,t_i)$ in the following lemma.

\begin{theorem}\label{thm:commit}
	Let $\alg^*$ be an $\alpha$-approximation algorithm for the commitment Pandora's box problem on boxes $F^*_1,\dots,F^*_n$. There exists an $\alpha$-approximation algorithm $\alg$ for the Pandora's box problem on $F_1,\dots,F_n$.
\end{theorem}
\begin{proof}
	First we define a sequence of boxes $F'_1,\dots,F'_n$, where for all $i$ we have $F'_i=(v_i,v_i-\min(v_i,\sigma_i),t_i)$.  That is, we set $v'_i = v_i$, $t'_i = t_i$, and $c'_i = v_i - \min(v_i, \sigma_i)$.  Note that by definition of $\sigma_i$ we have $\ex{v_i-\min(v_i,\sigma_i)\ |\ t_i}=\ex{c_i\ |\ t_i}$ for all $t_i$. Thus, by Lemma \ref{lm:approx} an (online) $\alpha$-approximation algorithm for the Pandora's box problem on $F'_1,\dots,F'_n$ implies an $\alpha$-approximation algorithm for the Pandora's box problem on $F_1,\dots,F_n$ as desired. Next, we construct $\alg'$ required by Lemma \ref{lm:approx} using $\alg^*$.  To construct $\alg'$, whenever $\alg^*$ attempts to open a box $F^*_i$, we open $F'_i$ and report $(v'_i-c'_i,0,t_i) = (\min(v_i,\sigma_i),0,t_i) = F^*_i$ to $\alg^*$. $\alg'$ keeps the same set of boxes as $\alg^*$. 
	
	Let $Y_i$ be a binary random variable that is $1$ if $\alg^*$ opens box $i$ and $0$ otherwise. Also, let $X_i$ be a binary random variable that is $1$ if $\alg^*$ keeps box $i$ and $0$ otherwise.
	Note that $v^*_i =\min(v_i,\sigma_i)\leq \sigma_i$, and $v^*_i$ achieves it maximum value whenever $v_i\geq \sigma_i$. In this case by the commitment constraint we have $X_i=Y_i$. Therefore we either have $v_i-\min(v_i,\sigma_i)=0$ or $X_i=Y_i$, which gives us
	\begin{align}\label{eq:commitYi}
		Y_i \big(v_i-\min(v_i,\sigma_i)\big) = X_i \big(v_i-\min(v_i,\sigma_i)\big)
	\end{align}
	 Then for any fixed profile of types, and taking expectations conditional on those type realizations, we have
	\begin{align*}
		u_{\alg'}&=\ex{\sum_{i=1}^n X_i v'_i - \sum_{i=1}^n Y_i c'_i} &\text{}\\
		&=\ex{\sum_{i=1}^n X_i v_i - \sum_{i=1}^n Y_i \big(v_i-\min(v_i,\sigma_i)\big)} &\text{By def.}\\
		&=\ex{\sum_{i=1}^n X_i v_i - \sum_{i=1}^n X_i \big(v_i-\min(v_i,\sigma_i)\big)}
		&\text{Eq. \eqref{eq:commitYi}}\\
		&=\ex{\sum_{i=1}^n X_i \Big( v_i - \big(v_i-\min(v_i,\sigma_i)\big)\Big)}&\\  
		&=\ex{\sum_{i=1}^n X_i \min(v_i,\sigma_i)}&\\  
		&=\ex{\sum_{i=1}^n X_i v^*_i}& \\
		&= u_{\alg^*},&  
	\end{align*}
	where the first equality is since $\alg'$ opens and keeps the same sets as $\alg^*$.
	Similarly, we can show $u_{\opt'}=u_{\opt^*}$, where $\opt'$ is the optimum algorithm for the Pandora's box problem on $F'_1,\dots,F'_n$ and $\opt^*$ is the optimum algorithm for the commitment Pandora's box problem on $F^*_1,\dots,F^*_n$. Therefore $\alg'$ is an $\alpha$-approximation algorithm for the Pandora's box problem on $F'_1,\dots,F'_n$ as promised.
\end{proof}

\subsection{A Reduction for Online Pandora's Box Problems}

In this section we use Theorem \ref{thm:commit} to provide a strong connection between the online Pandora's box problem under general constraints and prophet inequalities.  This leads to our main result: we prove that a threshold-based algorithm for a prophet inequality problem, under any given feasibility constraints on boxes that can be opened and/or prizes that can be kept, immediately translates into an algorithm for the online Pandora's box problem.  This reduction preserves the approximation factor of the threshold-based algorithm.  This implies several approximation algorithms for the online Pandora's box problem under different constraints, which we discuss in Section~\ref{sec:applications}.

Recall that, in the online Pandora's Box problem, the algorithm is permitted to keep a set $R\subseteq S$ of boxes and collect the reward, where $R$ and $S$ are restricted to be from arbitrary predefined collections of feasible collections $\mathbb{R}$ and $\mathbb{S}$. In the associated prophet inequality problem, the costs of all boxes are known to be $0$.  In Theorem~\ref{thm:reduction} below, we use the notion of {\em threshold-based algorithms} for the prophet inequality problem defined as follows. We say an algorithm is threshold-based if for every $i\in \{1,\dots,n\}$ we have a threshold $\tau_i(t_i)$ (where the threshold can depend on the type as well as the index) and the algorithm keeps a box if and only if it not less than $\tau_i(t_i)$. The threshold $\tau_i(t_i)$ may be adaptive, that is it may depend on any observation prior to observing the $i$th box.

\begin{theorem}\label{thm:reduction}
	Let $\alg_{\tau}$ be a threshold-based $\alpha$-approximation algorithm for the prophet inequalities problem, under a collection of constraints. There exists an $\alpha$-approximation algorithm for the online Pandora's box problem under the same constraints.
\end{theorem}
\begin{proof}
	We define $F^*_i = (v^*_i=\min(v_i,\sigma_i),c^*_i=0,t_i)$.
	Next we give an $\alpha$-approximation algorithm $\alg^*$ for the commitment Pandora's box problem on boxes $F^*_1,\dots,F^*_n$. This together with Theorem \ref{thm:commit} will prove the theorem.
	
	Let $\tau_i(.)$ be the threshold function used by $\alg_{\tau}$ given observed values $v_1^*,\ldots,v_n^*$. We define $\alg^*$ as follows. Upon arrival of box $F^*_i$, first we check if $\sigma_i\geq \tau_i(t_i)$. Note that this implicitly implies that the box is acceptable according to the constraints. If $\sigma_i\geq \tau_i(t_i)$ we open the box, otherwise we skip it. This ensures the commitment constraint. If we opened the box and $v^*_i\geq \tau_i$ we keep it, otherwise we ignore it and continue. It is easy to observe that $\alg_{\tau}$ and $\alg^*$ provide the same outcome and have the same approximation factor.
\end{proof}

\section{Algorithms For Online Pandora's Box via Prophet Inequalities}
\label{sec:applications}
Here we use the tools from the previous section to provide algorithms for the online Pandora's box problem under different kinds of constraints. First, as a warm-up, in Theorem \ref{thm:singleItem} we show a $1/2$-approximation algorithm for a simple version of the problem where there are no constraints on the set of boxes that we can open, and we can only keep the value of one box. Note that this is the online version of the classical Pandora's box problem.
Indeed it is known that there is no $1/2+\epsilon$ approximation algorithm for this problem even if all of the costs are $0$ (where the problem is equivalent to a basic prophet inequalities problem) \cite{KW-STOC12}.  We will prove Theorem~\ref{thm:singleItem} directly, without appealing to Theorem~\ref{thm:reduction}, to provide insight into how the given thresholds translate into a policy for the Pandora's Box problem.

\begin{theorem}\label{thm:singleItem}
	There exists a $1/2$-approximation algorithm for the online Pandora's box problem with no constraints on opening boxes, but the value of exactly one box is kept.
\end{theorem}
\begin{proof}
	We define $F^*_i = (v^*_i=\min(v_i,\sigma_i),c^*_i=0,t_i)$.
	Next we give a simple $1/2$-approximation algorithm $\alg^*$ for the commitment Pandora's box problem on boxes $F^*_1,\dots,F^*_n$. This together with Theorem \ref{thm:commit} proves the theorem. 
	
	Set $\tau$ such that $\prob{\max_{i=1}^n v^*_i \geq \tau}=\frac{1}{2}$. Let $j$ be a random variable that indicates the first index such that $v^*\geq \tau$. Let $v^*=v^*_j$, if there exists such index $j$, and let $v^*=0$ otherwise. It is known that $\ex{v^*}=\frac 1 2 \ex{\max_{i=1}^n v^*_i}$ \cite{samuel1984comparison}.
	
	We define $\alg^*$ as follows. Upon arrival of box $F^*_i$, first we check if $\sigma_i\geq \tau$. If it is we open the box, otherwise we skip it. This ensures the commitment constraint;  that is, if we observe $\sigma_i$, we will accept it. Next, if we opened the box and $v^*_i\geq \tau_i$ we keep it and terminate. Otherwise we continue to the next box. It is easy to observe that $\alg^*$ keeps $v^*$, and hence is a $1/2$-approximation algorithm.
\end{proof}

We now explore other applications of our reduction.  Theorem \ref{thm:reduction}, together with previously-known approximation algorithms for prophet inequality problems with various types of constraints, implies the existence of approximation algorithms for variations of the online Pandora's box problem.  Specifically, we have the following variations:
\begin{itemize}
\item {\em Online $k$-Pandora's box problem}: we are given a cardinality $k$, and at most $k$ boxes can be kept.
\item {\em Online knapsack Pandora's box problem}: we have a capacity $C$, and the type $t_i$ of each box corresponds to a size. The total size of the boxes that can be kept is at most $C$.
\item {\em Online matroid Pandora's box problem}: we have a matroid constraint on the set of boxes, and the boxes that are kept must be an independent set of the matroid.
\end{itemize}

We note that all of the variations above have no constraints on opening boxes;  however, in what follows, we study a variation of the problem with constraints on opening boxes.

As we have mentioned, prophet inequality approximation algorithms for the settings of cardinality constraints~\cite{alaei2014bayesian}, knapsack constraints~\cite{duetting2017prophet}, and matriod constraints~\cite{KW-STOC12} exist. Making use of these results and Theorem \ref{thm:reduction} implies the following corollary. 
\begin{corollary}
	There is
	\begin{itemize}
	\item a $1-\frac{1}{\sqrt{k+3}}$-approximation algorithm for the online $k$-Pandora's box problem,
	\item a $1/5$-approximation algorithm for the online knapsack Pandora's box problem,
	\item and a $1/2$-approximation algorithm for the online matroid Pandora's box problem.
	\end{itemize}
\end{corollary}

\section{Pandora's Box with Multiple Arms}
\label{sec:multiarm}

In the context reinforcement learning, we next consider a multi-arm
version of the Pandora's box problem. In each of $J$ rounds, $m$
boxes are presented.  There are therefore $mJ$ boxes in total.  At
most one box can be opened in each round.
The $m$ boxes presented in a given round are ordered; we can think of each box as having a $\emph{type}$ labeled $1$ through $m$.  All boxes of type $t$ have the same cost $c_t$, and also have the same value distribution $F_t$.  For notational convenience we'll write $v_{tj}$ for the value in the box of type $t$ presented at time $j$;  for convenience we assume $v_{tj} > 0$.  At the end of the $J$ rounds, the player can keep at most one prize for each type of box.  That is, if we write $S_t$ for the subset of boxes of type $t$ opened by the player, then the objective is to maximize
\[  \sum_t \max_{j \in S_t} v_{tj} - c_t \cdot |S_t|. \]
If no box of type $t$ is opened, we'll define $\max_{j \in S_t} v_{tj}$ (i.e., the prize for that type) to be 0.
This is a variant of the online Pandora's box problem with a (partition) matroid constraint on the set of prizes that can be kept, and also a constraint on the subsets of boxes that can be opened.

We can think of this problem as presenting boxes one at a time, where the boxes from a round are presented sequentially, with the additional constraint that one box per round can be selected.  Note that types here are not random, but would depend on the index of the box.  Our previous reduction applies, so that solving this multi-arm Pandora's box problem reduces to developing the related prophet inequality.
In this prophet inequality, boxes can be opened at no cost, but we must irrevocably choose whether or not to keep any given prize as it is revealed.  We can still keep at most one prize of each type, and we can still open at most one box in each time period.  Our question becomes: can we develop a threshold policy to achieve a constant-factor prophet inequality for this setting?  In this case, a threshold policy corresponds to choosing a threshold $\tau_t$ for each box type $t$, and accepting a prize $v_{tj}$ from an \emph{opened} box if and only if $v_{tj} > \tau_t$.

What is an appropriate benchmark for the prophet inequality?  Note that the sum of the best prizes of each type, ex post, might not be achievable by any policy due to the restriction on which boxes can be opened.  We will therefore compare against the following weaker benchmark.  We consider a prophet who must choose, in an online fashion, one box to open in each round, given knowledge of the prizes in previously opened boxes.  Then, after having opened one box on each of the $J$ rounds, the prophet can select the largest observed prize of each type.  In other words, the prophet has the advantage of being able to choose from among the opened boxes in retrospect, but must still open boxes in an online fashion.  Our goal is to obtain a constant approximation to the expected value enjoyed by such a prophet.

We begin with some observations about the choice of which box to open.  First, the optimal policy for the prophet is to open boxes greedily.  In particular, this policy can be implemented in time linear in $m$, each round.

\begin{lemma}
In each round $j$, it is optimal for the prophet to open a box of type $t$ that maximizes his expected value, as if the game were to end after time $j$.
\end{lemma}
\begin{proof}
	Note that the expected marginal gain of opening a box of type $t$ can only decrease over time, and only as more boxes of that type are opened.  Suppose $t$ is the box with maximum expected marginal value at time $j$.  Suppose further that the prophet does not open box $t$, and furthermore opens no box of type $t$ until the final round $J$.  Then $t$ will be the box with maximum expected marginal value in round $J$, and therefore it would be optimal to open the box of type $t$ on the last round.  This implies that it is optimal to open at least one box of type $t$, at some point between round $j$ and the end of the game.  The prophet is therefore at least as well off by opening the box of type $t$ immediately, since doing so does not affect the distribution of the revealed value, and this can only provides more information for determining which other boxes to open.  It is therefore (weakly) optimal to open box $t$ at time $j$.
\end{proof}

Similarly, once thresholds are fixed, an identical argument implies that the optimal threshold-based policy behaves greedily.  In particular, the optimal policy can be implemented in polynomial time, given an arbitrary set of thresholds.

\begin{lemma}
Suppose the player's policy is committed to selecting a prize of type $t$, from an opened box of type $t$, if and only if its value exceeds $\tau_t$.  Then, in each round $j$, it is optimal to open a box $t$, from among those types for which a prize has not yet been accepted, that maximizes his expected value as if the game were to end after time $j$.  That is, $t \in \arg\max_t  \{ \ex{v_{tj}\ |\ v_{tj} > \tau_t} \cdot \prob{v_{tj} > \tau_t} \}$.
\end{lemma}

We now claim that there are thresholds that yield a $2$-approximate prophet inequality for this setting.  First, some notation.  By the principle of deferred randomness, we can think of the value of the prize in any given box as only being determined at the moment the box is opened.  With this in mind, we will write $w_{tk}$ for the value observed in the $k$'th box of type $t$ opened by the decision-maker.  For example, $w_{t1}$ is the value contained in whichever box of type $t$ is opened first, regardless of the exact time at which it is opened.  Note that the behavior of any online policy is fully described by the profile of values $\mathbf{w} = (w_{tk})$, and each $w_{tk}$ is a value drawn independently from distribution $F_t$.  For such a profile $\mathbf{w}$, 
write $y^*_{tk}(\mathbf{w})$ for the indicator variable that is $1$ if the prophet opens at least $k$ boxes of type $t$, and keeps the $k$'th one opened.  Then the expected value enjoyed by the prophet is
\[ \ex[\mathbf{w}]{\sum_k \sum_{t} w_{tk} y^*_{tk}(\mathbf{w})}. \]
For each box type $t$, we will set the threshold
\[ \tau_t = \frac{1}{2} \ex[\mathbf{w}]{\sum_k w_{tk} y^*_{tk}(\mathbf{w})}. \]
That is, $\tau_t$ is half of the expected value obtained by the prophet from boxes of type $t$.  

To prove that these thresholds achieve a good approximation to the prophet's welfare, it will be useful to analyze the possible correlation between the number of boxes of type $t$ opened by the prophet, and whether any prize of type $t$ is kept by the threshold algorithm.  To this end, write $z_{tk}^*(\mathbf{w})$ for the indicator that the prophet opens at least $k$ boxes of type $t$.  Note that $z_{tk}^*(\mathbf{w}) \geq y_{tk}^*(\mathbf{w})$ for all $t$, $k$, and $\mathbf{w}$.  We'll also write $Q_t(\mathbf{w})$ for the indicator variable that is $1$ if $w_{tk} z_{tk}^*(\mathbf{w}) \leq \tau_t$ for all $k$.  That is, $Q_t(\mathbf{w}) = 1$ if no box of type $t$ opened by the prophet has value greater than $\tau_t$, and hence the threshold algorithm does not keep any prize of type $t$.  The following lemma shows that $Q_t$ is positively correlated with $z_{tk}^*$, for every $t$ and $k$.

\begin{lemma}
	\label{lem.cor}
	For all $t$ and $k$, $\ex[\mathbf{w}]{Q_t(\mathbf{w}) \cdot z^*_{tk}(\mathbf{w})} \geq \ex[\mathbf{w}]{Q_t(\mathbf{w})} \cdot \ex[\mathbf{w}]{z^*_{tk}(\mathbf{w})}$.
\end{lemma}
\begin{proof}
	Fix the values of $w_{\ell k}$ for all $\ell \neq t$.  Note that this also fixes the choice of which box the prophet will open, on any round in which the prophet chooses not to open a box of type $t$.  Due to the prophet's greedy method of opening boxes, at every time $k$, the prophet will choose to open the box of type $t$ if and only if the maximum value from a box of type $t$, seen so far (or $0$ if no box of type $t$ has been opened yet), is below a threshold determined by the values observed from the other boxes.  In other words, the values $w_{\ell k}$ for $\ell \neq k$ define a sequence of non-decreasing thresholds $h_0, h_1, h_2, \dotsc, h_J$ with the following property.  Suppose, at time $j$, the prophet has previously opened $s < j$ boxes of type $t$, and has therefore opened $j - s - 1$ boxes of types other than $t$.  Then the prophet will open box $t$ at time $j$ if and only if
	\begin{equation}
	\label{eq.thresh}
	\max_{r \leq s}\{ w_{tr} \} < h_{j-s-1}.
	\end{equation}
	Here and below, we'll take the maximum of an empty set to be $0$.
	
	We now claim that the prophet opens $k$ or more boxes of type $t$, at or before time $J$, if and only if $\max_{r < k} \{ w_{tr} \} < h_{J-k}$.  We prove this claim by induction on $J$.  The case $J = 1$ is immediate, since box $t$ is opened first if and only if $h_0 > 0$.  Now suppose $J > 1$.  If $\max_{r < k} \{w_{tr}\} \geq h_{J-k}$, then $\max_{r \leq k} \{w_{tr}\} \geq h_{j-k-1}$ at all times $j \leq J$, so by \eqref{eq.thresh} the prophet never chooses to open the $k$'th box of type $t$.  In the other direction, note that if $\max_{r < k} \{w_{tr}\} < h_{J-k}$, then $\max_{r < k-1} \{w_{tr}\} < h_{(J-1) - (k-1)}$, so by induction the prophet opens at least $k-1$ boxes by time $J-1$.  This means that either the prophet has already opened $k$ boxes of type $t$ before time $J$ (and we are done), or it has opened exactly $k-1$ boxes of type $t$ before time $J$.  In the latter case, since we have $\max_{r \leq k-1} \{w_{tr}\} < h_{J - (k-1) - 1}$, we conclude from \eqref{eq.thresh} that the prophet opens box $t$ at time $J$, as required.
	
	We are now ready to return to $z^*_{tk}$ and $Q_{t}$.  From the claim above, we have that $z^*_{tk}(\mathbf{w}) = 1$ if and only if $\max_{r < k} \{w_{tr}\} < h_{J-k}$.  It suffices to show that this event is only more likely to occur if we condition on the event $[ Q_t(\mathbf{w}) = 1 ]$.  We will actually consider a stronger event $A$, which is that $\max_{r < J} \{w_{tr}\} \leq \tau_t$.  Note that event $A$ is more stringent than the event $[ Q_t(\mathbf{w}) = 1 ]$, since event $A$ requires that all of the first $J$ values from box $t$ are at most $\tau_t$, whether or not those boxes are opened.  But since these events differ only on values that are in unopened boxes, we have $\ex{z^*_{tk}\ |\ A} = \ex{z^*_{tk}\ |\ Q_t}$, so it suffices to prove that $z^*_{tk}$ is positively correlated with event $A$.
	
	To show that $\ex{z^*_{tk}\ |\ A} \geq \ex{z^*_{tk}}$, we will couple outcomes with and without this conditioning on $A$.  To do so, we imagine first drawing a sequence $(w_{t1}, w_{t2}, \dotsc)$, then re-drawing any values that are greater than $\tau_t$ until all values are at most $\tau_t$; say $(w_{t1}', w_{t2}', \dotsc)$ is the modified profile.  Note that since $w_{ts}' \leq w_{ts}$ for all $s$, we have that if $[ \max_{r < k} \{w_{tr}\} < h_{J-k} ]$, then it must also be that $[ \max_{r < k} \{w_{tr}'\} < h_{J-k} ]$.  So the expected value of $z^*_{tk}$ can only increase as a result of this transformation.  Since $(w_{t1}, w_{t2}, \dotsc)$ is chosen uniformly from all profiles, and $(w_{t1}', w_{t2}', \dotsc)$ is distributed uniformly from profiles that satisfy event $A$, we conclude that
	$\ex{z^*_{tk}\ |\ A} \geq \ex{z^*_{tk}}$ as required.
\end{proof}

Finally, we can prove the multi-arm prophet inequality.

\begin{theorem}
\label{thm.multiarm}
	The optimal threshold policy, using the thresholds described above, achieves at least half of the expected value enjoyed by the prophet for the multi-armed prophet inequality problem.
\end{theorem}
\begin{proof}
	As above, write $y_{tk}^*(\mathbf{w})$ for the indicator variable for if the prophet opens at least $k$ boxes of type $t$ and keeps the $k$'th one opened, and write $z_{tk}^*(\mathbf{w})$ for the indicator that the prophet opens at least $k$ boxes of type $t$.  We'll show a $2$-approximation for the policy that uses the given thresholds, but chooses to open the same boxes that the prophet would open.  This will imply the theorem, since the optimal policy that uses these thresholds $\tau_t$ would do at least as well as the policy that opens the same boxes as the prophet.
	
	Given this choice of what boxes to open, let $y_{tk}(\mathbf{w})$ be the indicator variable that is $1$ if the threshold algorithm opens at least $k$ boxes of type $t$ and keeps the $k$'th one opened.  The total value obtained by the threshold algorithm is then
	\[ \ex[\mathbf{w}]{\sum_k \sum_{t} w_{tk} y_{tk}(\mathbf{w})}. \]
	Note that the threshold algorithm might not choose a prize of every type, since it might be that all observed prizes of type $t$ are less than $\tau_t$.
	Write $Q_{tk}(\mathbf{w})$ for the indicator variable that is $1$ if none of the first (up to) $k$ boxes of type $t$ opened by the threshold algorithm are kept by the algorithm.  We'll also write $Q_t(\mathbf{w})$ for the indicator variable that is $1$ if no prize of type $t$ is kept by the threshold algorithm at any time.  In particular, $Q_t(\mathbf{w}) \leq Q_{tk}(\mathbf{w})$ for all $k$.  Finally, we'll write
	$q_t = \ex[\mathbf{w}]{Q_t(\mathbf{w})}$.
	
	We decompose the value generated by the algorithm into (a) the value attributable to the thresholds, and (b) any value in excess of the thresholds.  That is,
	\begin{align*}
	&\ex[\mathbf{w}]{\sum_k \sum_{t} w_{tk} y_{tk}(\mathbf{w})} = \\
	&\ex[\mathbf{w}]{\sum_k \sum_{t} \tau_t y_{tk}(\mathbf{w})} + \ex[\mathbf{w}]{\sum_k \sum_{t} (w_{tk} - \tau_t) y_{tk}(\mathbf{w})}.
	\end{align*}
	For the first term, we have
	\begin{align}
	\label{eq.term1}
	&\ex[\mathbf{w}]{\sum_k \sum_{t} \tau_t y_{tk}(\mathbf{w})}
	= \sum_t \tau_t \ex[\mathbf{w}]{\sum_k y_{tk}(\mathbf{w})} \nonumber\\
	&= \sum_t \tau_t (1-q_t) \nonumber\\
	&= \frac{1}{2}\sum_t (1-q_t) \ex[\mathbf{w}]{\sum_k w_{tk} y^*_{tk}(\mathbf{w})}
	\end{align}
	The second term is more interesting.
	We have
	\begin{align}
	\label{eq.term2}
	&\ex[\mathbf{w}]{\sum_k \sum_{t} (w_{tk} - \tau_t) y_{tk}(\mathbf{w})}\nonumber\\
	& \geq \sum_t \sum_k \ex[\mathbf{w}]{ (w_{tk} - \tau_t)^+ \cdot z_{tk}^*(\mathbf{w}) \cdot Q_{tk}(\mathbf{w})} \nonumber\\
	& = \sum_t \sum_k \ex[\mathbf{w}]{(w_{tk} - \tau_t)^+} \cdot \ex[\mathbf{w}]{z_{tk}^*(\mathbf{w}) \cdot Q_{tk}(\mathbf{w})} \nonumber\\
	& \geq \sum_t \sum_k \ex[\mathbf{w}]{ (w_{tk} - \tau_t)^+} \cdot \ex[\mathbf{w}]{z_{tk}^*(\mathbf{w}) \cdot Q_{t}(\mathbf{w})} \nonumber\\
	& \geq \sum_t \sum_k \ex[\mathbf{w}]{ (w_{tk} - \tau_t)^+} \cdot \ex[\mathbf{w}]{z_{tk}^*(\mathbf{w})} \cdot \ex[\mathbf{w}]{ Q_{t}(\mathbf{w}) } \nonumber\\
	& = \sum_t q_t \sum_k \ex[\mathbf{w}]{ (w_{tk} - \tau_t)^+ \cdot z_{tk}^*(\mathbf{w})} \nonumber\\
	& \geq \sum_t q_t \sum_k \ex[\mathbf{w}]{ (w_{tk} - \tau_t)^+ \cdot y_{tk}^*(\mathbf{w})} \nonumber\\
	& \geq \sum_t q_t \left( \ex[\mathbf{w}]{\sum_k w_{tk} y_{tk}^*(\mathbf{w})} - \tau_t \right) \nonumber\\
	& = \frac{1}{2} \sum_t q_t \ex[\mathbf{w}]{ \sum_k w_{tk} y_{tk}^*(\mathbf{w}) },
	\end{align}
	where the first inequality is linearity of expectation and the definition of $y$, the equality on the second line uses the fact that $w_{tk}$ is independent of the values that occur earlier (which determine $z_{tk}^*$ and $Q_{tk}$, and the inequality on the fourth line is Lemma~\ref{lem.cor}.
	The result now follows by adding \eqref{eq.term1} and \eqref{eq.term2}, yielding
	\begin{align*}
	\ex[\mathbf{w}]{\sum_k \sum_{t} w_{tk} y_{tk}(\mathbf{w})}
	&\geq \frac{1}{2} \sum_t \ex[\mathbf{w}]{\sum_k w_{tk} y_{tk}^*(\mathbf{w})}\\
	&=  \frac{1}{2} \ex[\mathbf{w}]{\sum_k \sum_t w_{tk} y_{tk}^*(\mathbf{w})}
	\end{align*}
	as required.
\end{proof}

We note that Theorem~\ref{thm.multiarm} is constructive, and the thresholds can be computed to within an arbitrarily small error in polynomial time.  For example, this can be done by sampling instances of $\mathbf{w}$, simulating the behavior of the prophet, and calculating an empirical average of the associated threshold values.  See~\cite{duetting2017prophet} for further details on this approach.

Applying the reduction from Theorem~\ref{thm:reduction} to the prophet inequality in Theorem~\ref{thm.multiarm} yields a polynomial time algorithm for the multi-arm online Pandora's box problem.  

\begin{corollary}
There is an $2$-approximation algorithm for the online multi-arm Pandora's box problem.
\end{corollary}

\section{Conclusion and Open Problems}
\label{sec:conclusion}

We have presented a general reduction method for translating
approximation algorithms in the prophet inequality setting to
corresponding approximation algorithms in the online Pandora box
setting with applications to information learning. Further we have
introduced a novel multi-armed bandit Pandora box variation in the
context of reinforcement learning where our methods apply. Along the
way, we have considered many generalizations of the Pandora box
problem, including allowing distributions on costs and the use of
types. 

One open challenge is to relax the assumption that values, costs, 
and/or types of different boxes are independent.  One could also
generalize to objectives beyond maximizing the sum of the values 
selected minus the cost of opening boxes.  For example, what 
happens if the cost of opening a box increases as more boxes are opened?
Finally, the multi-arm prophet inequality is a special case of a more general
class of stochastic optimization problems, and it would be interesting to
extend to more general scenarios.  For example, can one extend the result
to distributions that vary across time, or to general matroid constraints over 
the set of boxes that can be opened?

\section*{Acknowledgments}
Hossein Esfandiari was supported in part by NSF grants CCF-1320231 and CNS-1228598.

MohammadTaghi HajiAghayi was supported in part by NSF CAREER award CCF-1053605, NSF AF:Medium grant CCF-1161365, NSF BIGDATA grant IIS-1546108, NSF SPX grant CCF-1822738, UMD AI in Business and Society Seed Grant and UMD Year of Data Science Program Grant.
Part of this work was done while visiting Microsoft Research New England.

Michael Mitzenmacher was supported in part by NSF grants CCF-1563710,
CCF-1535795, CCF-1320231, and CNS-1228598; also, part of this work was done while visiting Microsoft Research New England.

\bibliographystyle{aaai}
\bibliography{Bibliography}

\begin{thebibliography}{}

\bibitem[\protect\citeauthoryear{Alaei, Hajiaghayi, and
  Liaghat}{2012}]{alaei2012online}
Alaei, S.; Hajiaghayi, M.; and Liaghat, V.
\newblock 2012.
\newblock Online prophet-inequality matching with applications to ad
  allocation.
\newblock In {\em ACM EC},  18--35.

\bibitem[\protect\citeauthoryear{Alaei}{2014}]{alaei2014bayesian}
Alaei, S.
\newblock 2014.
\newblock Bayesian combinatorial auctions: Expanding single buyer mechanisms to
  many buyers.
\newblock {\em SIAM Journal on Computing} 43(2):930--972.

\bibitem[\protect\citeauthoryear{Babaioff, Immorlica, and
  Kleinberg}{2007}]{babaioff2007matroids}
Babaioff, M.; Immorlica, N.; and Kleinberg, R.
\newblock 2007.
\newblock Matroids, secretary problems, and online mechanisms.
\newblock In {\em SODA},  434--443.

\bibitem[\protect\citeauthoryear{Esfandiari \bgroup et al\mbox.\egroup
  }{2017}]{esfandiari2015prophet}
Esfandiari, H.; Hajiaghayi, M.; Liaghat, V.; and Monemizadeh, M.
\newblock 2017.
\newblock Prophet secretary.
\newblock {\em SIAM Journal on Discrete Mathematics} 31(3):1685--1701.

\bibitem[\protect\citeauthoryear{Feldman, Svensson, and
  Zenklusen}{2015}]{FSZ-SODA15}
Feldman, M.; Svensson, O.; and Zenklusen, R.
\newblock 2015.
\newblock {A simple O (log~log~(rank))-competitive algorithm for the matroid
  secretary problem}.
\newblock In {\em SODA},  1189--1201.

\bibitem[\protect\citeauthoryear{Gittins and Jones.}{1974}]{GJ74}
Gittins, J., and Jones., D.
\newblock 1974.
\newblock A dynamic allocation index for the sequential design of experiments.
\newblock In {\em Progress in Statistics},  241--266.

\bibitem[\protect\citeauthoryear{Goel and Mehta}{2008}]{GM-SODA08}
Goel, G., and Mehta, A.
\newblock 2008.
\newblock Online budgeted matching in random input models with applications to
  adwords.
\newblock In {\em SODA},  982--991.

\bibitem[\protect\citeauthoryear{Guruganesh and Singla.}{2017}]{GS-IPCO17}
Guruganesh, G.~P., and Singla., S.
\newblock 2017.
\newblock Online matroid intersection: Beating half for random arrival.
\newblock In {\em IPCO},  241--253.

\bibitem[\protect\citeauthoryear{Hajiaghayi, Kleinberg, and
  Parkes}{2004}]{hajiaghayi2004adaptive}
Hajiaghayi, M.~T.; Kleinberg, R.; and Parkes, D.~C.
\newblock 2004.
\newblock Adaptive limited-supply online auctions.
\newblock In {\em ACM EC},  71--80.

\bibitem[\protect\citeauthoryear{Karande, Mehta, and
  Tripathi}{2011}]{KMT-STOC11}
Karande, C.; Mehta, A.; and Tripathi, P.
\newblock 2011.
\newblock Online bipartite matching with unknown distributions.
\newblock In {\em STOC},  587--596.

\bibitem[\protect\citeauthoryear{Kesselheim \bgroup et al\mbox.\egroup
  }{2013}]{kesselheim2013optimal}
Kesselheim, T.; Radke, K.; T{\"o}nnis, A.; and V{\"o}cking, B.
\newblock 2013.
\newblock An optimal online algorithm for weighted bipartite matching and
  extensions to combinatorial auctions.
\newblock In {\em ESA},  589--600.

\bibitem[\protect\citeauthoryear{Kleinberg and Weinberg}{2012}]{KW-STOC12}
Kleinberg, R., and Weinberg, S.~M.
\newblock 2012.
\newblock Matroid prophet inequalities.
\newblock In {\em STOC},  123--136.

\bibitem[\protect\citeauthoryear{Korula and P{\'a}l}{2009}]{KorulaPal-ICALP09}
Korula, N., and P{\'a}l, M.
\newblock 2009.
\newblock {Algorithms for secretary problems on graphs and hypergraphs}.
\newblock In {\em ICALP}.
\newblock  508--520.

\bibitem[\protect\citeauthoryear{Krengel and
  Sucheston}{1977}]{krengel1977semiamarts}
Krengel, U., and Sucheston, L.
\newblock 1977.
\newblock Semiamarts and finite values.
\newblock {\em Bulletin of the American Mathematical Society}.

\bibitem[\protect\citeauthoryear{Krengel and
  Sucheston}{1978}]{krengel1978semiamarts}
Krengel, U., and Sucheston, L.
\newblock 1978.
\newblock On semiamarts, amarts, and processes with finite value.
\newblock {\em Advances in Prob} 4:197--266.

\bibitem[\protect\citeauthoryear{Lachish}{2014}]{Lachish-FOCS14}
Lachish, O.
\newblock 2014.
\newblock O (log log rank) competitive ratio for the matroid secretary problem.
\newblock In {\em FOCS},  326--335.

\bibitem[\protect\citeauthoryear{Mahdian and Yan}{2011}]{MY-STOC11}
Mahdian, M., and Yan, Q.
\newblock 2011.
\newblock Online bipartite matching with random arrivals: an approach based on
  strongly factor-revealing lps.
\newblock In {\em STOC},  597--606.

\bibitem[\protect\citeauthoryear{Mohammad~Taghi~Hajiaghayi and
  Sandholm.}{2007}]{hajiaghayi2007automated}
Mohammad~Taghi~Hajiaghayi, R.~K., and Sandholm., T.
\newblock 2007.
\newblock Automated online mechanism design and prophet inequalities.
\newblock In {\em AAAI},  58--65.

\bibitem[\protect\citeauthoryear{Paul~Duetting and
  Lucier.}{2017}]{duetting2017prophet}
Paul~Duetting, Michal~Feldman, T.~K., and Lucier., B.
\newblock 2017.
\newblock Prophet inequalities made easy: Stochastic optimization by pricing
  non-stochastic inputs.
\newblock In {\em FOCS},  540--551.

\bibitem[\protect\citeauthoryear{Robert~Kleinberg and Weyl.}{2016}]{KWW-16}
Robert~Kleinberg, B.~W., and Weyl., E.~G.
\newblock 2016.
\newblock Descending price optimally coordinates search.
\newblock In {\em ACM EC},  23--24.

\bibitem[\protect\citeauthoryear{Samuel-Cahn}{1984}]{samuel1984comparison}
Samuel-Cahn, E.
\newblock 1984.
\newblock Comparison of threshold stop rules and maximum for independent
  nonnegative random variables.
\newblock In {\em the Annals of Probability},  12(4):1213--1216.

\bibitem[\protect\citeauthoryear{Singla.}{2018}]{Singla18}
Singla., S.
\newblock 2018.
\newblock The price of information in combinatorial optimization.
\newblock In {\em SODA},  2523--2532.

\bibitem[\protect\citeauthoryear{Weitzman.}{1979}]{weitzman79}
Weitzman., M.~L.
\newblock 1979.
\newblock Optimal search for the best alternative.
\newblock In {\em Econometrica},  47(3):641--654.

\bibitem[\protect\citeauthoryear{Yan}{2011}]{yan2011mechanism}
Yan, Q.
\newblock 2011.
\newblock Mechanism design via correlation gap.
\newblock In {\em SODA},  710--719.

\end{thebibliography}
\end{document}